%% file: clc2012_submission_6_final.tex
\documentclass[preliminary,copyright,creativecommons]{eptcs}
\providecommand{\event}{CL\&C 2012}

\usepackage{amsfonts}
\usepackage{amsmath}
\usepackage{amssymb}
\usepackage{amsthm}
\usepackage{bussproofs}

\input{preamble0}

\title{Applying G\"{o}del's Dialectica Interpretation to Obtain a Constructive Proof of Higman's Lemma}
\author{Thomas Powell
\institute{Queen Mary, University of London\\ United Kingdom}
\email{tpowell@eecs.qmul.ac.uk}}
\def\titlerunning{A Constructive Proof of Higman's Lemma}
\def\authorrunning{Thomas Powell}

\begin{document}
\maketitle

\begin{abstract}We use G\"{o}del's Dialectica interpretation to analyse Nash-Williams' elegant but non-constructive `minimal bad sequence' proof of Higman's Lemma. The result is a concise constructive proof of the lemma (for arbitrary decidable well-quasi-orders) in which Nash-Williams' combinatorial idea is clearly present, along with an explicit program for finding an embedded pair in sequences of words.\end{abstract}

\section{Introduction}
\label{sec-intro}

We call a preorder $(\wq,\leq_\wq)$ a \textit{well-quasi-order} (WQO) if any infinite sequence $(\eseq_i)$ has the property that $\eseq_i\leq_\wq \eseq_j$ for some $i<j$. The theory of WQOs contains several results which state that certain constructions on WQOs inherit well-quasi-orderedness, the most famous being Kruskal's tree theorem \cite{Kruskal(60)}. A special case of this theorem is Higman's lemma:

\begin{theorem}[Higman, \cite{Higman(52)}]\label{thm-higman}If $(\wq,\leq_\wq)$ is a WQO, then so is the set $(\wwq,\leq_\wwq)$ of words in $\wq$ under the embeddability relation $\leq_\wwq$, where $\pair{\e_0,\ldots,\e_{m-1}}\leq_{\wwq}\pair{\es_0,\ldots,\es_{n-1}}$ iff there is a strictly increasing map $f\colon\is{m}\to\is{n}$ with $\e_i\leq_\wq \es_{fi}$ for all $i<m$.\end{theorem}

A short proof of Higman's lemma (and more generally Kruskal's theorem) was given by Nash-Williams \cite{Nash-Williams(63)}, using an elegant but non-constructive combinatorial idea known as the \textit{minimal bad sequence} argument.

Higman's lemma has attracted a great deal of attention in logic and computer science, and has been a focal point of research into computational aspects of classical reasoning used in infinitary combinatorics. The constructive content of Nash-Williams' minimal bad sequence argument has been widely analysed (see for instance \cite{Coquand(91),Veldman(2004)}), and in particular, constructive content has been extracted from the proof using  formal methods such the $A$-translation \cite{Murthy(90)} and inductive definitions \cite{Coquand(93)}. An extensive study of program extraction for Higman's lemma has been carried out by Berger and Seisenberger (see \cite{BS(2005),Seisenberger(2003)}), who improve the aforementioned techniques and implement them in the {\sc Minlog} system.

In this article we give another constructive proof of Higman's lemma based on the minimal bad sequence argument. The novelty of our approach is that we use a technique that has not been applied in this context - G\"{o}del's \emph{Dialectica} interpretation. The combination of the negative translation and the Dialectica interpretation forms an extremely powerful and efficient method for extracting programs from classical proofs - testament to this is its central role in the well-known \emph{proof mining} program (see \cite{Kohlenbach(2008)}).

The formal extraction of computational information from proofs often results in output that is complex, highly syntactic and difficult to understand in \emph{mathematical} terms. However, the use of proof theoretic techniques to analyse the constructive content of classical reasoning is becoming increasingly relevant in mathematics, therefore we believe that it is important to produce case studies in which these techniques are applied in a transparent and intuitive manner.

The goal of this article is not just a new proof of Higman's lemma, but a case study that sheds some light on the functional interpretation of proofs in infinitary combinatorics. Our emphasis here is not on `mining' the proof for quantitative information but to produce a constructive justification of Higman's lemma that can actually be read as a mathematical proof, and in which Nash-Williams' original combinatorial idea is clearly present. In addition, we give a heuristic account of the operational behaviour of the resulting program.

\subsection{Preliminaries}
\label{subsec-introduction-prelim}

We formalise Higman's lemma in the language $\PAomega$ of Peano arithmetic in all finite types (see e.g. \cite{Avigad(98)} for details), although throughout the paper we endeavour to avoid \emph{excessive} formality and make various syntactic shortcuts to keep things as readable as possible. By extending $\PAomega$ with the axiom of \emph{dependent choice}
\begin{equation*}\DC \ \colon \ \forall n,x^X\exists y^X\; A_n(x,y)\to\forall x_0\exists f^{\NN\to X}(f(0)=x_0\wedge \forall n\; A_n(fn,f(n+1)))\end{equation*}
over arbitrary types $X$, one obtains a theory of analysis capable of formalising a large portion of mathematics, including Nash-Williams' minimal bad sequence construction. \\

\noindent\textbf{Notation.} We make use of the following conventions and abbreviations.
\begin{itemize}

\item $0_X$ denotes a canonical element of type $X$.

\item Because we will be confronted with a large number of variables, we often use the convention that when a term of type $X$ is denoted $x$, sequences of terms of the same type will often be denoted in bold type ${\bf x}$.

\item $s\ast\alpha$ represents the concatenation of the finite sequence $s$ and a finite/infinite sequence $\alpha$.

\item We write $s\sbs \alpha$ when the finite sequence $s$ is an initial segment of a finite/infinite sequence $\alpha$.

\item $\initSeg{\alpha}{n}$ is the initial segment of the infinite sequence $\alpha$ of size $n$.

\item We write $a\sbw b$ when a word $a\colon X^\ast$ is an initial segment of $b$ i.e. $|a|\leq |b|$ and $a_i=b_i$ for all $i<|a|$. If $a$ is a \emph{prefix} ($|a|<|b|$) we write $a\psbw b$.

\item Given two sequences of words $\wseq$ and $\wseqs$ we write $\wseq\smin{n}\wseqs:\equiv(\initSeg{\wseq}{n}=\initSeg{\wseqs}{n}\wedge\wseq_n\sbw\wseqs_n)$ and $\wseq\psmin{n}\wseqs:\equiv(\initSeg{\wseq}{n}=\initSeg{\wseqs}{n}\wedge\wseq_n\psbw\wseqs_n)$ - the latter simply states that $u$ is lexicographically less than $v$ at point $n$ with respect to the prefix relation $\psbw$.

\end{itemize}

\subsection{The functional interpretation of proofs in $\PAomega+\DC$}
\label{subsec-introduction-dialectica}

This article assumes familiarity with G\"{o}del's functional interpretation of classical proofs, by which we mean the Dialectica interpretation combined with the negative translation. We do not have space to give details of the interpretation - for this the reader is referred to \cite{Avigad(98)}. However, it is useful to recall a few basic facts.
\begin{itemize}

\item The functional interpretation of $\Sigma_2$ formulas coincides with the well-known \emph{no-counterexample interpretation} of Kreisel, interpreting $A\equiv\exists x\forall y A_0(x,y)$ as a functional $F$ that witnesses $\forall f\exists x A(x,fx)$. Intuitively $F$ justifies $A$ by refuting arbitrary counterexample functions $f$ attempting to disprove $A$.

\item The functional interpretation interprets $\Pi_2$ formulas $\forall x\exists y B(x,y)$ \emph{directly} with a functional $f$ satisfying $\forall x B(x,fx)$, due to the fact that it admits Markov's principle. This means that we can use the interpretation to extract programs from even \emph{classical} proofs of $\Pi_2$ theorems.

\end{itemize}

It was shown by G\"{o}del that $\PAomega$ has a functional interpretation in the system $\systemT$ of higher-type primitive recursive functionals. On the other hand, system $\systemT$ is insufficient to interpret the combination of classical logic and countable choice. For this, one typically assigns a direct realizer to the negative translation of choice, usually some form of backward induction such as the well-known \emph{bar recursion} devised by Spector in \cite{Spector(62)}. In this article dependent choice is interpreted using the more recent \emph{product of selection functions} introduced in \cite{EO(2009)}.

\begin{definition}\label{defn-prod}A \emph{selection function} is any functional of type $\sft{X}{R}:\equiv(X\to R)\to X$, for arbitrary $X$, $R$. Given an indexed family of selection functions $\varepsilon\colon X^\ast\to \sft{X}{R}$ together with functionals $q\colon X^\omega\to R$ and $\varphi\colon X^\omega\to \NN$, the product of selection functions $\EPSs$ is defined by the recursion schema
\begin{equation*}\label{eqn-eps}\EPS{\varphi}{\varepsilon}{q}{s}\stackrel{X^\omega}{:=}\left\{\begin{array}{ll}{\bf 0}_{\wq^\omega} & \mbox{  if $\varphi(\ext{s})<|s|$} \\[2mm] a_s\ast\EPS{\varphi}{\varepsilon}{q_{a_s}}{s\ast a_s} & \mbox{  otherwise}\end{array}\right. \end{equation*}
where $a_s=\varepsilon_s(\lambda x\; . \; q_x(\EPS{\varphi}{\varepsilon}{q_x}{s\ast x}))$, $q_x$ is defined by $q_x(\alpha):=q(x\ast\alpha)$ and $\hat{s}$ is the canonical extension of $s$.\end{definition}

$\EPSs$ is a variant of bar recursion that makes explicit the idea that bar recursion can be viewed as kind of backtracking algorithm analogous to the computation of optimal strategies in games of \emph{unbounded} length. We feel it is good practise to choose it over Spector's original bar recursion because it comes naturally equipped with this game semantics. The idea is to imagine $q\colon X^\omega\to R$ specifying the outcome of a sequential game with moves of type $X$ and outcome of type $R$, the $\varepsilon_s$ as \emph{selection functions} that specify a strategy for round $|s|$ given that $s$ has already been played and $\varphi\colon X^\omega\to \NN$ as a control functional that indicates when the game has terminated. For further details on the $\EPSs$ see \cite{EO(2011A)}. By unwinding Definition \ref{defn-prod} one can prove the following key result.

\begin{theorem}[Main theorem on $\EPSs$, cf. \cite{OP(2012A)}]\label{thm-mainspector}Setting $\alpha{:=}\EPS{\varphi}{\varepsilon}{q}{\pair{}}$ and $p_s:=\lambda x\; . \; q_{s\ast x}(\EPS{\varphi}{\varepsilon}{q_{s\ast x}}{s\ast x})$ solves the following system of equations
\begin{equation}\label{eqn-spector}\begin{aligned}\alpha n &\stackrel{X}{=} \varepsilon_{\initSeg{\alpha}{n}}(p_{\initSeg{\alpha}{n}}) \\
q(\alpha) &\stackrel{R}{=} p_{\initSeg{\alpha}{n}}(\alpha n)\end{aligned} \end{equation}
for all $n\leq\varphi\alpha$.\end{theorem}

As originally established by Spector, in order to witness the functional interpretation of dependent choice it is sufficient to solve the equations (\ref{eqn-spector}) given $\varepsilon$, $q$ and $\varphi$. Therefore a consequence of Theorem \ref{thm-mainspector} is that $\EPSs$ realizes the functional interpretation of dependent choice. For full details of the interpretation of choice via $\EPSs$ the reader is referred to \cite{OP(2012A)}. In this article however, it is enough to know that $\EPSs$ solves (\ref{eqn-spector}) - in our interpretation of the minimal bad sequence construction an instance of these equations naturally arises and we will solve them directly using $\EPSs$, bypassing the formal interpretation of choice.

The statement that $\wwq$ is a WQO can be written as a $\Pi_2$ sentence. By formalising the classical proof of Higman's lemma in $\PAomega+\DC$, we guarantee in theory that given a realizer for the well-quasi-orderedness of $\wq$ we can extract a direct realizer $\alg\colon (\wwq)^\omega\to\NN$ in $\systemT+\EPSs$ that bounds the search for an embedded pair in an arbitrary sequence of words. We formalise the proof in Sect. \ref{sec-formal} and extract a realizer $\alg$ in Sect. \ref{sec-extract}.

\section{A Classical Proof of Higman's Lemma}
\label{sec-classical}

We begin by presenting Nash-Williams' proof of Higman's lemma. First we need the following simple result.

\begin{lemma}\label{lem-ramsey}In a WQO $(X,\leq_X)$, any sequence $(x_i)$ has an infinite increasing subsequence.\end{lemma}

\begin{proof}For general WQOs this is an easy consequence of Ramsey's theorem.\end{proof}

In the following we call a sequence in a preorder $X$ \emph{good} if $x_i\leq_X x_j$ for some $i<j$. A sequence is \emph{bad} if it is not good. $X$ is a WQO if all sequences in $X$ are good.

\begin{proof}[Proof of Theorem \ref{thm-higman} (Nash-Williams, \cite{Nash-Williams(63)})] Suppose for contradiction that $X$ is a WQO, but there exists at least one bad sequence $\wseq$ in $(\wwq)^\omega$. Then among all bad sequences we pick a \textit{minimal} bad sequence as follows:
\begin{enumerate}
\item Choose $\wseqs_0$ to be an element of $\wwq$ with the property that $\wseqs_0$ is the first element of some bad sequence but no prefix of $\wseqs_0$ extends to a bad sequence in this way. Such an element exists by the assumption that we have at least one bad sequence $\wseq$.
\item Given that $\wseqs_0,\ldots,\wseqs_{n-1}$ have been selected, choose $\wseqs_n$ to be an element with the property that $\wseqs_0,\ldots,\wseqs_n$ starts a bad sequence but $\wseqs_0,\ldots,\wseqs_{n-1},\ws$ does not extend to a bad sequence for any prefix $\ws\psbw\wseqs_n$.
\end{enumerate}
By dependent choice we can construct an infinite sequence $(\wseqs_i)$ in this manner. It is easy to see that $(\wseqs_i)$ must itself be bad and therefore in particular each word $\wseqs_i$ must be non-empty, so we can write $\wseqs_i=\twseqs_i\ast\eseqs_i$ where the $\eseqs_i$ form an infinite sequence in $\wq$.

Now by Lemma \ref{lem-ramsey} the sequence $(\eseqs_i)$ has an increasing subsequence $$\eseqs_{i_0}\leq_\wq\eseqs_{i_1}\leq_\wq\ldots.$$ Consider the sequence $$\wseqs_0,\ldots,\wseqs_{i_0-1},\twseqs_{i_0},\twseqs_{i_1},\ldots.$$ This sequence must be bad, else $(\wseqs_i)$ would be good, but $\twseqs_{i_0}$ is a proper initial segment of $\wseqs_{i_0}$, contradicting the minimality of $(\wseqs_i)$ at $i_0$. Therefore there cannot exist an initial bad sequence $u$ in $\wwq$.\end{proof}

\section{Formalising the Classical Proof}
\label{sec-formal}

We now formalise Nash-Williams' proof in $\PAomega+\DC$, so that we are ready to apply the functional interpretation in the next section. Given a preorder $(\wq,\leq_\wq)$ define the predicate $\bad_\wq$ on $\wq^\omega\times\NN$ by
\begin{equation*}\label{defn-bad}\bad_\wq(\eseq,j):\equiv\forall i_0<i_1\leq j(x_{i_0}\nleq_\wq x_{i_1}).\end{equation*}
We define the predicate $\bad_\wwq$ on $(\wwq)^\omega\times\NN$ similarly. We suppress the subscript on $\bad$ when it is clear which type it applies to.

\begin{remark}In this article the intuition is that the underlying WQO $\wq$ consists of elements of type $0$, and that the relation $\leq_\wq$ is decidable. Therefore $\sbw$, $\sbs$, $\smin{n}$ and $\bad$ will all be decidable over both $\wq$ and $\wwq$.\end{remark}

A sequence $\eseq$ is bad is it satisfies the $\Pi_1$ predicate $\forall j\bad(\eseq,j)$. The preorder $X$ is a WQO if the closed $\Pi_2$ predicate $\WQO{\wq}:\equiv\forall \eseq\exists j\neg\bad_\wq(\eseq,j)$ holds, similarly $\wwq$ is a WQO if $\WQO{\wwq}:\equiv\forall\wseq\exists j\neg\bad_\wwq(\wseq,j)$ holds. Higman's lemma can then be formally written as
\begin{equation*}\WQO{\wq}\to\WQO{\wwq}.\end{equation*}
In the proof of Higman's lemma, the hypothesis $\WQO{\wq}$ appears in the form given by Lemma \ref{lem-ramsey}, namely that any sequence in $\wq$ has an infinite monotone subsequence:
\begin{equation}\label{defn-ram}\Ram{\wq} :\equiv \forall\eseq^{\wq^\omega}\exists \func^{\NN\to\NN}\forall k \forall i<j\leq k(\func i<\func j\wedge \eseq_{\func i}\leq_\wq\eseq_{\func j}). \end{equation}
In our interpretation of Nash-Williams' proof we do not analyse the computational content of Lemma \ref{lem-ramsey}, rather we directly interpret
\begin{equation*}\Ram{\wq}\to\WQO{\wwq}.\end{equation*}
There are two reasons for this - the first is that in general the passage from $\WQO{\wq}$ to $\Ram{\wq}$ requires Ramsey's theorem and therefore full dependent choice, so while one could in theory interpret Lemma \ref{lem-ramsey} using bar recursion or the product of selection functions, in this article we wish to focus on the main content of Nash-William's proof, so we omit these details.

The second reason is that in certain interesting cases it is easy to prove $\Ram{\wq}$ directly, without resorting to Ramsey's theorem. For instance, when the underlying alphabet $\wq$ is a finite set, $\Ram{\wq}$ is provable in $\PAomega$ using the infinite pigeonhole principle, and so a realizer for the functional interpretation of $\Ram{\wq}$ can be given in system $\systemT$.

\subsection{The Minimal Bad Sequence Argument}
\label{subsec-formal-mbs}

Our main step in the formalisation of Nash-Williams' proof is the formalisation of his minimal bad sequence argument. The main non-trivial principle of $\PAomega$ we require is the \emph{least element principle} -
\begin{equation*}\label{def-LEP}\LEP \ \colon \ \exists mA(m)\to\exists m'(A(m')\wedge\neg A(m'-1)),\end{equation*}
where in our version we assume that $A$ is monotone in the sense that it satisfies $(i) \ i<j\to (A(i)\to A(j))$ and $(ii) \ \neg A(0)$.

\begin{lemma}[Minimal bad sequence construction]\label{lem-mbsformal}It it provable in $\PAomega+\DC$ that for any sequence of words $u\colon (\wwq)^\omega$, there exists a sequence $\seqext_u\equiv \seqext^0,\seqext^1,\ldots$ of sequences of type $(\wwq)^\omega$ and a sequence $\seqfunc_u\equiv\seqfunc^0,\seqfunc^1,\ldots$ of functions of type $(\wwq)^\omega\to\NN$, which, defining $\seqext^{-1}:=u$, together satisfy the following sentences:
\begin{align}\label{nest}& \forall n(\initSeg{\seqext^{n-1}}{n}=\initSeg{\seqext^n}{n}); \\
&\label{badimp} \forall n,j(\neg\bad(\seqext^n,j)\to\neg\bad(\seqext^{n-1},j));\\
&\label{minimality} \forall n,q^{(\wwq)^\omega}(q\psmin{n}\seqext^n\to\neg\bad(q,\seqfunc^nq)).\end{align}
\end{lemma}

This formulation of the minimal bad sequence construction is a little more intricate than that given in Sect. \ref{sec-classical}, in particular our aim is to highlight the computational aspects of the construction. The intuition is that the sequence $\seqext_u$ is \textit{classically} constructed in the following manner:
\begin{enumerate}
\item Given an initial sequence $u$, we choose $\seqext^0$ to be a bad sequence such that $\seqext^0_0\sbw u_0$ but no $y\psbw \seqext^0_0$ extends to a bad sequence. If no prefix of $u_0$ extends to a bad sequence we set $\seqext^0:=u$.
\item Given that we have constructed $\seqext^{n-1}$, we choose $\seqext^n$ to be a bad extension of $\initSeg{\seqext^{n-1}}{n}$ such that $\initSeg{\seqext^{n}}{n}\ast y$ does not extend to a bad sequence for any $y\psbw\seqext^n_n$. If no such bad extension exists, we set $\seqext^{n}:=\seqext^{n-1}$.
\end{enumerate}
If $\seqext_u$ is defined in this way then it clearly satisfies (\ref{nest}), and for each $\seqext^n$ we can produce a (classically constructed) function $\seqfunc^n$ that witnesses the minimality of $\seqext^n$ in the sense of (\ref{minimality}).

We observe that the $\seqext^n$ are not necessarily bad (in fact if $\wq$ is a WQO they never will be), but the point is that $\seqext^n$ only fails to be bad in the event that $\seqext ^{n-1}$ is good, in which case we must have $\seqext^n=\seqext^{n-1}$. This is the intuition behind (\ref{badimp}). Nash-Williams' proof is based on the fact that if $\wq$ is a WQO then by (\ref{minimality}) we can show that there is some $n$ and $j$ such that $\bad(\seqext^n,j)$ fails, and then by induction over (\ref{badimp}) we must have $\neg\bad(u,j)$.

\begin{proof}[Proof of Lemma \ref{lem-mbsformal}] Suppose for the moment that $n$ and $w^{(\wwq)^\omega}$ are fixed. Define the monotone predicate $A(m):\equiv\exists r^{(\wwq)^\omega}\forall i\nd{A_m}{r}{i}$ where
\begin{equation*}\label{defn-a}\nd{A_m}{r}{i}:\equiv r\smin{n} w\wedge |r_n|<m\wedge(\bad(w,i)\to\bad(r,i)).\end{equation*}
It is clear that $A(m)$ is monotone, and that $\forall i\nd{A_{|w_n|+1}}{w}{i}$ holds. Therefore by $\LEP$ there exists some $m'$ such that
\begin{equation}\label{eqn-lep1}\left\{\begin{aligned}& \exists p\forall j\left(p\smin{n} w\wedge |p_n|<m'\wedge\left(\bad(w,j)\to\bad(p,j)\right)\right) \wedge \\ & \forall q\exists k\left(q\smin{n}w\wedge |q_n|<m'-1\to\left(\bad(w,k)\wedge\neg\bad(q,k)\right)\right) \end{aligned}\right . .\end{equation}
Now, observing that if $p\smin{n} w\wedge |p_n|<m'$ then $q\psmin{n} p\to q\smin{n} w\wedge |q_n|<m'-1$ we can prove in $\PAomega$ that (\ref{eqn-lep1}) implies
\begin{equation}\label{eqn-lep2}\exists p\left(\forall j\left(\initSeg{w}{n}=\initSeg{p}{n}\wedge\left(\bad(w,j)\to\bad(p,j)\right)\right)\wedge\forall q\exists k\left(q\psmin{n} p\to\neg\bad(q,k)\right)\right).\end{equation}
Skolemizing (\ref{eqn-lep2}) we have that for arbitrary $n$, $w$, there exists a sequence $p$ and function $f\colon (\wwq)^\omega\to\NN$ satisfying
\begin{equation}\label{eqn-dcprem}\forall j,q\left(\initSeg{w}{n}=\initSeg{p}{n}\wedge\left(\bad(w,j)\to\bad(p,j)\right)\wedge\left(q\psmin{n} p\to\neg\bad(q,fq)\right)\right).\end{equation}
By $\DC$ of type $(\wwq)^\omega\times((\wwq)^\omega\to\NN)$ applied to (\ref{eqn-dcprem}) (only dependent on the sequence part of the previous entry), defining an initial value $\seqext^{-1}:=u$ there exists an infinite sequence of sequences $\seqext_u\equiv\seqext^0,\seqext^1\ldots$ and functions $\seqfunc_u\equiv\seqfunc^0,\seqfunc^1\ldots$ satisfying
\begin{equation}\label{eqn-dcconc}\forall n,j,q (\initSeg{\seqext^{n-1}}{n}=\initSeg{\seqext^n}{n}\wedge\left(\bad(\seqext^{n-1},j)\to\bad(\seqext^n,j)\right) \wedge\left(q\psmin{n} \seqext^n\to\neg\bad(q,\seqfunc^n q)\right) ). \end{equation}

This completes the proof, as (\ref{nest}), (\ref{badimp}) and (\ref{minimality}) clearly follow from (\ref{eqn-dcconc}).\end{proof}

In the following $\MB{\wwq}$ abbreviates the statement that for all $u$ there exists $\seqext_u$ and $\seqfunc_u$ satisfying (\ref{eqn-dcconc}).

\subsection{Completing the Proof}
\label{subsec-formal-completing}

\noindent\textbf{Notation.} Given a non-empty word $x\colon\wwq$ we write $x=\ft{x}\ast\lt{x}$ where $\ft{x}\colon\wwq$ and $\lt{x}\colon\wq$. So that these are well defined for all $x$, we define $\ft{\pair{}}:=\pair{}$ and $\lt{\pair{}}=0_X$.
Given a sequence of $\seqext\colon ((\wwq)^\omega)^\omega$ we define the diagonal sequences $\fts{\seqext}\colon(\wwq)^\omega$ by $\fts{\seqext}_i:=\ft{\seqext^i_i}$ and $\lts{\seqext}\colon\wq^\omega$ by $\lts{\seqext}_i:=\lt{\seqext^i_i}$.

\begin{theorem}\label{thm-higmanformal}It is provable in $\PAomega$ that $\Ram{\wq}\wedge\MB{\wwq}\to\WQO{\wwq}$.\end{theorem}

\begin{proof}Take an arbitrary sequence $u\colon(\wwq)^\omega$. By $\MB{\wwq}$ there exists $\seqext_u$ and $\seqfunc_u$ satisfying (\ref{nest}-\ref{minimality}). We show that one of the $\seqext^i$ must be good, which by (\ref{badimp}) implies that $u$ must also be good.

By $\Ram{\wq}$ applied to $\lts{\seqext}$ there exists a monotone function $g$ such that $\lts{\seqext}_{gi}\leq_\wq\lts{\seqext}_{gj}$ for all $i<j$. Define
\begin{equation*}\label{defn-psi}\psi\stackrel{(\wwq)^\omega}{:=}\initSeg{\seqext^{g0-1}}{g0}\ast (\fts{\seqext}_{gi})_{i\in\NN}\equiv\seqext^{g0-1}_0,\ldots,\seqext^{g0-1}_{g0-1},\fts{\seqext}_{g0},\fts{\seqext}_{g1},\ldots \end{equation*}
Now either $\seqext^{g0}_{g0}$ is empty (and hence $\seqext^{g0}$ is trivially good) or $\fts{\seqext}_{g0}\psbw\seqext^{g0}_{g0}$ and thus $\psi\psmin{g0}\seqext^{g0}$, which by (\ref{minimality}) implies that $\neg\bad(\psi,\seqfunc^{g0}\psi)$ i.e. the sequence
\begin{equation*}\initSeg{\psi}{\seqfunc^{g0}\psi+1}\equiv \seqext^{g0-1}_{g0},\ldots,\seqext^{g0-1}_{g0-1},\fts{\seqext}_{g0},\ft{\seqext}_{g1},\ldots,\fts{\seqext}_{g(\seqfunc^{g0}\psi-g0)}\end{equation*}
has one word contained in a later one. But by construction of $g$ this implies that the sequence
\begin{equation*}\seqext^{g0-1}_{0},\ldots,\seqext^{g0-1}_{g0-1},{\seqext^{g0}_{g0}},\seqext^{g0+1}_{g0-1},\ldots,\seqext^{g(\seqfunc^{g0}\psi-g0)}_{g(\seqfunc^{g0}\psi-g0)},\seqext^{g(\seqfunc^{g0}\psi-g0)+1}_{g(\seqfunc^{g0}\psi-g0)+1} \ \ (\ast)\end{equation*}
has one element contained in a later one (note that $\ft{x}\leq_\wwq\ft{y}\to x\leq_\wwq y$ unless $|x|=1$ and $|y|=0$, which is why we need to add the extra element at the end of $(\ast)$). But by the nesting property $(\ast)$ is just an initial segment of $\seqext^{g(\seqfunc^{g0}\psi-g0)+1}$, which must therefore be good. This completes the proof. \end{proof}

Combining Theorem \ref{thm-higmanformal} with Lemma \ref{lem-mbsformal} we see that $\Ram{\wwq}\to\WQO{\wwq}$ can be formalised in $\PAomega+\DC$. The proof as a whole is illustrated in Fig. \ref{fig-classical}.

\begin{figure}[t]
\begin{center}
{\footnotesize
\begin{prooftree}
\AxiomC{$\LEP$}
\AxiomC{$\DC$}
\doubleLine
\RightLabel{\scriptsize{Lem. \ref{lem-mbsformal}}}
\BinaryInfC{$\MB{\wwq}$}
\AxiomC{}
\doubleLine
\RightLabel{\scriptsize{Thm. \ref{thm-higmanformal}}}
\UnaryInfC{$\Ram{\wq}\wedge\MB{\wwq}\to\WQO{\wwq}$}
\BinaryInfC{$\Ram{\wq}\to\WQO{\wwq}$}
\end{prooftree}
}
\end{center}
\caption{Structure of Nash-Williams' proof.}
\label{fig-classical}
\end{figure}

\subsection{Computational Aspects of Nash-Williams' Proof}
\label{subsec-formal-completing}

Now that we have formalised Nash-Williams', we pause for a moment before the full program extraction to look at the computational hints contained in the classical proof. Assuming a realizer $g$ for $\Ram{\wq}$, given an arbitrary sequence of words $u\colon(\wwq)^\omega$ suppose we construct $\seqext_u$, $\seqfunc_u$ as in Lemma \ref{lem-mbsformal} and the sequence $\psi$ as in the proof of Theorem \ref{thm-higmanformal}.

By inspecting the proof of Theorem \ref{thm-higmanformal}, it is not too difficult to show that there exists $i_0<i_1\leq\phi(u)$ such that $u_{i_0}\leq_\wwq u_{i_1}$, where
\begin{equation*}\phi(u):=g(\seqfunc^{g0}_u\psi)+1.\end{equation*}
To see this, note that we prove that $\neg\bad(\seqext^{g(\seqfunc^{g0}\psi-g0)+1},g(\seqfunc^{g0}\psi-g0)+1)$ and so therefore we also have $\neg\bad(u,g(\seqfunc^{g0}\psi-g0)+1)$ by (\ref{badimp}) and hence $\neg\bad(u,\phi(u))$ since $g$ is monotone.

Now $\phi(u)$ is clearly an \emph{ineffective} bound for Higman's lemma, as it depends on non-constructive objects $g$, $\seqext_u$ and $\seqfunc_u$. However, in order to verify the correctness of $\phi(u)$, we do not need the whole of these objects. Rather
\begin{itemize}

\item $g$ must satisfy (\ref{defn-ram}) up to $k=\seqfunc^{g0}\psi$,
\item $\seqext_u$, $\seqfunc_u$ must satisfy (\ref{nest}-\ref{minimality}) up to $n=\phi(u)$.

\end{itemize}

Therefore, if we have a procedure that will compute \emph{approximations} to these objects up to a finite point parametrised by those objects themselves, we can turn $\phi$ into an \emph{effective} bound for Higman's lemma. This is precisely what the functional interpretation does.

\section{A Constructive Proof of Higman's Lemma}
\label{sec-extract}

We now build our constructive version of Nash-Williams' proof. This section follows closely the structure of Sect. \ref{sec-formal}. Recall that we assume a realizer for the functional interpretation of $\Ram{\wq}$, namely a functional $G\colon \wq^\omega\to((\NN^\NN\to\NN)\to(\NN\to\NN))$ satisfying (cf. (\ref{defn-ram}))
\begin{equation}\label{eqn-ramnd}\forall x^{\wq^\omega},\varphi^{\NN^\NN\to\NN}\forall i<j\leq\varphi(G^x_\varphi)(G^x_\varphi < G^x_\varphi j\wedge x_{G^x_\varphi i}\leq_\wq x_{G^x_\varphi j}).\end{equation}
In general, such a realizer could be obtained from a realizer of $\WQO{\wq}$ by implementing a computational interpretation of Ramsey's theorem - such as the one given in \cite{OP(2011B)} using the product of selection functions. However, when $\wq$ is finite, $G$ can be given directly using the standard interpretation of the infinite pigeonhole principle found in e.g. \cite{Oliva(2006)}.

\subsection{Interpreting the Minimal Bad Sequence Argument}
\label{subsec-extract-mbs}

The central part of our constructive proof is the following, constructive version of Lemma \ref{lem-mbsformal}, which is just a realizer for the functional interpretation of $\MB{\wwq}$. \\

\noindent\textbf{Notation.} Recall (Sect. \ref{subsec-introduction-dialectica}) that we denote the type of a selection function by $\sft{X}{R}:\equiv (X\to R)\to X$. We use the abbreviation $Y\equiv (\wwq)^\omega\times ((\wwq)^\omega\to\NN)$ for the type of our choice sequence. Also, in what follows it will be useful to implicitly write variables $F\colon A\to B\times C$ as pairs $\pair{F_0^{A\to B},F_1^{A\to C}}$ - this slight abuse of types will make our syntax much more intuitive.

\begin{lemma}[Minimal bad sequence construction]\label{lem-mbsextract}For fixed $n$ and $w^{(\wwq)^\omega}$ define the decidable formula $\nd{A^{n,w}_m}{r}{i}$ by
\begin{equation*}\label{defn-a-s}\nd{A^{n,w}_m}{r}{i}:=r\smin{n} w\wedge |r_n|<m\wedge\bad(r,i),\end{equation*}
which is slightly simpler than that used in the proof of Lemma \ref{lem-mbsformal}\footnote{It would have been sufficient, although less direct, to obtain (\ref{eqn-lep2}) in the proof of Lemma \ref{lem-mbsformal} by applying $\LEP$ to this simpler formula. We opt for this variant now to simplify the subsequent constructions, as either version would result in essentially the same program.}. Define the functionals $$\varepsilon_{n,w}=\pair{\varepsilon_{n,w}^0,\varepsilon_{n,w}^1}\colon \sft{Y}{\NN\times (\wwq)^\omega}$$ by
\begin{equation}\label{defn-selection}\pair{\varepsilon_{n,w}^0{\pair{J^{Y\to\NN},Q^{Y\to (\wwq)^\omega}}},\varepsilon_{n,w}^1\pair{J,Q}}\stackrel{Y}{:=}\pair{p_i,f_i}\end{equation}
where $i\leq |w_n|$ is the greatest integer satisfying $\neg\nd{A^{n,w}_{i}}{Q(p_{i},f_{i})}{f_{i}(Q(p_{i},f_{i}))}$ and the finite sequences $p_0,\ldots,p_{|w_n|}$ and $f_0,\ldots,f_{|w_n|}$ are defined recursively by
\begin{equation}\label{defn-fx}\begin{aligned}  f_0 &\stackrel{}{:=} 0_{(\wwq)^\omega\to\NN} \\  f_{i} &:= \lambda q.J(q,f_{i-1}) \\ p_{|w_n|} &\stackrel{}{:=} w \\  p_{i-1} &:= Q(p_i,f_{i}).\end{aligned}\end{equation}
Now, given an arbitrary sequence $u\colon (\wwq)^\omega$, define the family of selection functions $\tilde\varepsilon^u\colon Y^\ast\to \sft{Y}{\NN\times (\wwq)^\omega}$ by
\begin{equation}\label{defn-selections}\tilde\varepsilon^u_{\pair{P,F}}\pair{J,Q}:=\varepsilon_{|\pair{P,F}|,P^{|\pair{P,F}|-1}}\pair{J,Q},\end{equation}
where we define the initial value $P^{-1}:=u$. Now, given counterexample functionals $\Omega,\Phi\colon Y^\omega\to\NN$ and $\Psi\colon Y^\omega\to (\wwq)^\NN$, the sequences
\begin{equation*}\seqext_u,\seqfunc_u \stackrel{Y^\omega}{:=}\EPS{\Omega}{\tilde\varepsilon^u}{\pair{\Phi,\Psi}}{\pair{}}\end{equation*}
satisfy, defining $\seqext^{-1}_u:=u$, the following sentences (cf. (\ref{nest}-\ref{minimality})):
\begin{align}& \label{ndnest}\forall n\leq\Omega_{\seqext,\seqfunc}(\initSeg{\seqext^{n-1}}{n}=\initSeg{\seqext^n}{n}); \\
&\label{badimpnd} \forall n\leq\Omega_{\seqext,\seqfunc}(\neg\bad(\seqext^n,\Phi_{\seqext,\seqfunc})\to\neg\bad(\seqext^{n-1},\Phi_{\seqext,\seqfunc}));\\
&\label{minimalitynd} \forall n\leq\Omega_{\seqext,\seqfunc}(\Psi_{\seqext,\seqfunc}\psmin{n}\seqext^n\to\neg\bad(\Psi_{\seqext,\seqfunc},\seqfunc^n(\Psi_{\seqext,\seqfunc}))).\end{align}\end{lemma}

These sequences $\seqext_u$, $\seqfunc_u$ computed via the product of selection functions interpret the instance of $\DC$ used in the minimal bad sequence construction, and witness the no-counterexample interpretation of $\MB{\wwq}$. The functional $\Omega$ determines how large the approximation to the choice sequence is, and $\Phi$, $\Psi$ in some sense calibrate its \emph{depth}.

Our aim in the next section is to pick suitable counterexample functions such that (\ref{minimalitynd}) implies $\neg\bad(\seqext^n,\Phi_{\seqext,\seqfunc})$ for some $n\leq\Omega_{\seqext,\seqfunc}$, then by induction over (\ref{badimpnd}) we have $$\neg\bad(\seqext^n,\Phi_{\seqext,\seqfunc})\to\neg\bad(\seqext^{-1},\Phi_{\seqext,\seqfunc}))\equiv\neg\bad(u,\Phi_{\seqext,\seqfunc}),$$ and we therefore obtain $\exists i_0<i_1\leq\Phi_{\seqext_u,\seqfunc_u}(u_{i_0}\leq_\wwq u_{i_1})$ i.e. a constructive bound for $u$ being good. First we must prove the lemma.

\begin{proof}[Proof of Lemma \ref{lem-mbsextract}] First, we show that $\varepsilon_{n,w}$ witnesses the functional (i.e. no-counterexample) interpretation of (\ref{eqn-dcprem}), in the sense that given counterexample functions $J,Q\colon Y\to\NN\times (\wwq)^\omega$ for $j,q$ we have (suppressing dependencies and writing $\varepsilon^b\stackrel{Y}{=}\varepsilon^b_{n,w}(\pair{J,Q})$)
\begin{equation}\label{eqn-dcpremnd}\initSeg{w}{n}=\initSeg{\varepsilon^0}{n}\wedge (\bad(w,J\varepsilon)\to\bad(\varepsilon^0,J\varepsilon))\wedge (Q\varepsilon\psmin{n} \varepsilon^0\to\neg\bad(Q\varepsilon,\varepsilon^1(Q\varepsilon))).\end{equation}
The following is a constructive version of the proof of Lemma \ref{lem-mbsformal}. Let $0\leq i\leq |w_n|$ be the greatest number such that $\neg\nd{A^{n,w}_{i}}{Q(p_{i},f_{i})}{f_{i}(Q(p_{i},f_{i}))}$, so by definition we have $\pair{\varepsilon^0,\varepsilon^1}={\pair{p_i,f_i}}$. There are two cases. \\

\noindent\emph{Case 1: $i=|w_n|$}. Then we have
\begin{equation*}\neg\nd{A_{|w_n|}}{Q\varepsilon}{\varepsilon^1(Q\varepsilon))}\equiv Q\varepsilon\smin{n} w\wedge |(Q\varepsilon)_n|<|w_n|\to\neg\bad(Q\varepsilon,\varepsilon^1(Q\varepsilon)).\end{equation*}
Therefore, observing that $\varepsilon^0=p_{|w_n|}:=w$ and $(Q\varepsilon)\psmin{n}\varepsilon^0\to (Q\varepsilon)\smin{n} w\wedge |(Q\varepsilon)_n|<|w_n|$, we easily obtain (\ref{eqn-dcpremnd}). \\

\noindent\emph{Case 2: $i<|w_n|$}. By maximality of $i$, $\nd{A_{i+1}}{Q(p_{i+1},f_{i+1})}{f_{i+1}(Q(p_{i+1},f_{i+1}))}$ must be true. Now looking at the defining equations (\ref{defn-fx}), we have $Q(p_{i+1},f_{i+1})=p_i=\varepsilon^0$ and $f_{i+1}(Q(p_{i+1},f_{i+1}))=f_{i+1}(p_i)=J(p_i,f_{i})=J\varepsilon$, therefore the following two formulas are true:
\begin{align}\label{eqn-lepeqa}\nd{A_{i+1}}{\varepsilon^0}{J\varepsilon} &\equiv \varepsilon^0\smin{n} w\wedge |(\varepsilon^0)_n|\leq i\wedge \bad(\varepsilon^0,J\varepsilon); \\ \label{eqn-lepeqb}
\neg\nd{A_i}{Q\varepsilon}{\varepsilon^1(Q\varepsilon)} &\equiv Q\varepsilon\smin{n} w\wedge |(Q\varepsilon)_n|<i\to \neg\bad(Q\varepsilon,\varepsilon^1(Q\varepsilon)). \end{align}
Now by (\ref{eqn-lepeqa}) we have $\initSeg{w}{n}=\initSeg{\varepsilon^0}{n}\wedge (\bad(w,J\varepsilon)\to\bad(\varepsilon^0,J\varepsilon))$, and because $Q\varepsilon\psmin{n}\varepsilon^0\to Q\varepsilon\smin{n} w\wedge |(Q\varepsilon)_n|<i$ by (\ref{eqn-lepeqb}) we obtain $Q\varepsilon\psmin{n}\varepsilon^0\to\neg\bad(Q\varepsilon,\varepsilon^1(Q\varepsilon))$. Therefore (\ref{eqn-dcpremnd}) holds.

Thus we have shown that $\varepsilon_{n,w}$ witnesses (\ref{eqn-dcpremnd}) for arbitrary $n,w,J$ and $Q$. Now setting
\begin{equation}\label{defn-JQ}\begin{aligned}\seqext_u,\seqfunc_u &\stackrel{Y^\omega}{:=} \EPS{\Omega}{\tilde\varepsilon^u}{\pair{\Phi,\Psi}}{\pair{}}
\\ J_n(p,f) &\stackrel{\NN}{:=} \Phi_{\pair{\initSeg{\seqext_u}{n},\initSeg{\seqfunc_u}{n}}\ast\pair{p,f}}(\EPS{\Omega}{\tilde\varepsilon^u}{\pair{\Phi,\Psi}}{\pair{\initSeg{\seqext_u}{n},\initSeg{\seqfunc_u}{n}}\ast\pair{p,f}}) \\ Q_n(p,f) &\stackrel{(\wwq)^\omega}{:=}\Psi_{\pair{\initSeg{\seqext_u}{n},\initSeg{\seqfunc_u}{n}}\ast\pair{p,f}}(\EPS{\Omega}{\tilde\varepsilon^u}{\pair{\Phi,\Psi}}{\pair{\initSeg{\seqext_u}{n},\initSeg{\seqfunc_u}{n}}\ast\pair{p,f}})\end{aligned}\end{equation}
by the main theorem on $\EPSs$ quoted in Sect. \ref{subsec-introduction-dialectica} we satisfy Spector's equations
\begin{equation}\label{eqn-spectors}\begin{aligned}\seqext^n,\seqfunc^n &= \varepsilon^0_{n,\seqext^{n-1}}(J_n,Q_n),\varepsilon^1_{n,\seqext^{n-1}}(J_n,Q_n) \\  J_n(\seqext^n,\seqfunc^n),Q_n(\seqext^n,\seqfunc^n) &= \Phi_{\seqext,\seqfunc},\Psi_{\seqext,\seqfunc}\end{aligned}\end{equation}
for all $n\leq\Omega_{\seqext,\seqfunc}$. By setting $w:=\seqext^{n-1}$, $J:=J_n$ and $Q:=Q_n$ in (\ref{eqn-dcpremnd}) and substituting in (\ref{eqn-spectors}), we obtain equations (\ref{ndnest}-\ref{minimalitynd}).\end{proof}

\subsection{Constructing a Realizer for Higman's Lemma}
\label{subsec-extract-realizer}

\begin{definition}\label{defn-higmanextract}Given a pair of sequences $\seqext\colon((\wwq)^\omega)^\omega$ and $\seqfunc\colon ((\wwq)^\omega\to\NN)^\omega$, let $G_{\seqext,\seqfunc}$ be a realizer for $\Ram{\wq}$ on the sequence $(\lts{\seqext}_i)$ and counterexample function
\begin{equation*}\label{defn-higomega}\varphi_{\seqext,\seqfunc}:=\lambda g\; . \; \seqfunc^{g0}(\initSeg{\seqext^{g0-1}}{g0}\ast(\fts{\seqext}_{gi})).\end{equation*}
Define the functionals $\Omega$, $\Phi$ and $\Psi$ by (suppressing the subscript on $G$, $\varphi$)
\begin{equation*}\label{defn-higctrex}\begin{aligned}\Omega(\seqext,\seqfunc) &:= G(\varphi G)+1, \\
\Phi(\seqext,\seqfunc) &:= G(\varphi G)+1, \\
\Psi(\seqext,\seqfunc) &:= \initSeg{\seqext^{G0-1}}{G0}\ast(\fts{\seqext}_{Gi}).\end{aligned}\end{equation*}
Finally, define $\alg\colon (\wwq)^\omega\to\NN$ by
\begin{equation*}\label{defn-higrealizer}\alg(u):=\Phi(\seqext_u,\seqfunc_u),\end{equation*}
where $\seqext_u,\seqfunc_u:=\EPS{\Omega}{\tilde\varepsilon^u}{\pair{\Phi,\Psi}}{\pair{}}$ with $\tilde\varepsilon^u$ defined as in Lemma \ref{lem-mbsextract}. \end{definition}

The main theorem of this article is the following, constructive analogue of Theorem \ref{thm-higmanformal}.

\begin{theorem}[Higman's lemma, constructive version]\label{thm-higmanextract}Suppose $\wq$ is a WQO. Then for all sequences of words $u\colon\wwq$ over $\wq$ we have
\begin{equation*}\label{eqn-higman}\exists i_0<i_1\leq\alg(u)(u_{i_0}\leq_\wwq u_{i_1})\end{equation*}
where $\Gamma$ is constructed as in Definition \ref{defn-higmanextract}.\end{theorem}

\begin{proof}Fix $u$. In what follows, $\seqext$, $\seqfunc$ are fixed as $\seqext_u$, $\seqfunc_u$. We use the abbreviation $\Omega_u:=\Omega(\seqext_u,\seqfunc_u)$, and similarly for $\Phi_u$, $\Psi_u$, $G_u$ and $\varphi_u$. We claim that there is some $n\leq\Omega_u$ satisfying $\neg\bad(\seqext^n,\Phi_u)$. Then by induction over (\ref{badimpnd}), we see that $\neg\bad(u,\Phi_u)$, and the theorem follows from the definition of $\bad$. It remains to prove the claim.

First observe that because $G_u$ is a realizer of $\Ram{\wq}$ for $\varphi_u$ we have (cf. \ref{eqn-ramnd})
\begin{equation}\label{eqn-higram}\forall i<j\leq\varphi_u(G_u)(G_ui<G_uj\wedge \lts{\seqext}_{G_ui}\leq_\wq \lts{\seqext}_{G_uj}).\end{equation}
Now, $G0\leq G(\varphi G)$ so we have $G0< G(\varphi G)+1=\Omega_u$, therefore by (\ref{minimalitynd}) it follows that
\begin{equation}\label{eqn-higmin}\Psi_u\psmin{G0}\seqext^{G0}\to\neg\bad(\Psi_u,\seqfunc^{G0}(\Psi_u)).\end{equation}
The premise of (\ref{eqn-higmin}) must hold by construction of $\Psi_u$, since $\initSeg{\seqext^{G0-1}}{G0}=\initSeg{\seqext^{G0}}{n}$ by (\ref{ndnest}) and $\lts{\seqext}_{G0}\psbw \seqext^{G0}_{G0}$ (unless $\seqext^{G0}_{G0}=\pair{}$ in which case we trivially have $\neg\bad(\seqext^{G0},G0+1)$ and hence $\neg\bad(\seqext^{G0},\Phi_u)$). Therefore we have $\neg\bad(\Psi_u,\varphi G)$ since $\seqfunc^{G_u0}(\Psi_u)=\varphi_uG_u$ by definition, i.e. the finite sequence
\begin{equation*}\initSeg{\Psi_u}{\varphi G+1}\equiv\seqext^{G0-1}_0,\seqext^{G0-1}_1,\ldots,\seqext^{G0-1}_{G0-1},\fts{\seqext}_{G0},\ldots,\fts{\seqext}_{G(\varphi G-G0)}\end{equation*}
has one element contained in a later one (we illustrate the case $\varphi G\geq G0$ - if $\varphi G<G0$ then $\initSeg{\seqext^{G0-1}}{G0}$ is bad and hence $\neg\bad(\seqext^{G0-1},\Phi_u)$). Now since $\varphi G-G0\leq\varphi G$, by (\ref{eqn-higram}) we see that the sequence
\begin{equation*}\seqext^{G0-1}_{0},\seqext^{G0-1}_1,\ldots,\seqext^{G0-1}_{G0-1},{\seqext^{G0}_{G0}},\seqext^{G0+1}_{G0+1},\ldots,\seqext^{G(\varphi G-G0)}_{G(\varphi G-G0)},\seqext^{G(\varphi G-G0)+1}_{G(\varphi G-G0)+1} \ \ (\ast)\end{equation*}
has one element contained in a later one (we need to add an extra element for the same reason as we do in the proof of Theorem \ref{thm-higmanformal}). But because $G(\varphi G-G0)+1\leq G(\varphi G)+1=\Omega_u$, by the nesting property (\ref{ndnest}) the sequence $(\ast)$ is just an initial segment of $\seqext^{G(\varphi G-G0)+1}$, and hence $\neg\bad(\seqext^{G(\varphi G-G0)+1},G(\varphi G-G0)+1)$ which implies $\neg\bad(\seqext^{G(\varphi G-G0)+1},\Phi_u)$. This proves the claim, completing the proof.\end{proof}

An rough map of our constructive proof, with partial realizers shown is given as Fig. \ref{fig-constructive}.

\begin{figure}[t]
\begin{center}
{\footnotesize
\begin{prooftree}
\AxiomC{$\tilde\varepsilon\colon\LEP$}
\AxiomC{$\EPSs\colon\DC$}
\doubleLine
\RightLabel{\scriptsize{Lem. \ref{lem-mbsextract}}}
\BinaryInfC{$\EPSs(\tilde\varepsilon)\colon\MB{\wwq}$}
\AxiomC{}
\doubleLine
\RightLabel{\scriptsize{Thm. \ref{thm-higmanextract}}}
\UnaryInfC{$\Ram{\wq}\wedge\MB{\wwq}\to\WQO{\wwq}$}
\BinaryInfC{$\lambda G\; . \; \lambda u\; . \; \Phi^G(\EPS{\Omega}{\tilde\varepsilon}{\pair{\Phi,\Psi}}{\pair{}}\colon\Ram{\wq}\to\WQO{\wwq}$}
\end{prooftree}
}
\end{center}
\caption{Structure of constructive proof.}
\label{fig-constructive}
\end{figure}

\subsection{An Informal Discussion on the Extracted Program $\alg$}
\label{subsec-extract-behaviour}

We conclude the section with an \emph{informal} analysis of our extracted realizer. Often, programs extracted from classical proofs via proof interpretations can be very difficult to understand, sometimes taking up several pages of abstruse higher type syntax or computer code to even state. In contrast, given the logical complexity of Nash-Williams' proof our realizer extracted using the Dialectica interpretation is relatively concise, and we can even describe its operational behaviour to an extent.

\emph{We stress that everything which follows is heuristic and has not been properly formalised.} Our aim is merely to illustrate that it is at least feasible to decipher our realizer on a qualitative level!

Our algorithm uses the product of selection functions $\EPSs$ to interpret the minimal bad sequence argument used in Nash-Williams' proof. As observed in Sect. \ref{subsec-introduction-dialectica}, $\EPSs$ - and consequently our extracted program - comes equipped with a natural game theoretic semantics. For a full account of this the reader is advised to consult \cite{EO(2011A),OP(2012A)}. However, for completeness we state, without further details, the game theoretic reading of the key constructions in our algorithm.
\begin{itemize}

\item The functionals $\Phi,\Psi$ assign to any sequence (i.e. infinite play) $\seqext,\seqfunc$ an \emph{outcome} of type $\NN\times(\wwq)^\omega$.

\item The selection functions $\tilde\varepsilon^u$ - built from the realizer of $\LEP$ - implement a strategy for constructing an \emph{optimal} play $\seqext_u,\seqfunc_u$, the selection function $\tilde\varepsilon^u_{n,\seqext^{n-1}}$ being responsible for constructing the $n$th point $\seqext_u^n,\seqfunc_u^n$ in the sequence given that we have already computed the previous value $\seqext_u^{n-1}$.

\item The selection functions make a decision based on the functionals $J_n,Q_n$ defined in (\ref{defn-JQ}) which (in loose game theoretic terms) describe the optimal outcome of each potential choice at point $n$.

\item The functional $\Omega$ acts as a control, determining the `relevant part' of an infinite play $\seqext,\seqfunc$ thereby telling $\EPSs$ when it has computed a sufficiently long sequence.

\end{itemize}
In terms of Nash-Williams proof, the sequence $\seqext_u,\seqfunc_u$ strategically constructed by $\EPSs$ constitutes an `attempt' at producing a minimal bad sequence from $u$ (given by $\seqext$, with accompanying functionals $\seqfunc^n$ witnessing minimality at point $n$). We define $\Phi$, $\Psi$ and $\Omega$ so that the construction can be essentially reversed to obtain a bound for $u$.

So what can we say about this optimal sequence $\seqext_u,\seqfunc_u$? We prove in Theorem \ref{thm-higmanextract} that there is some element of the approximation $\seqext_u^n$ such that $\neg\bad(\seqext_u^n,\Phi_u)$ holds. It is not too difficult to see, by (\ref{eqn-spectors}), that $\neg\bad(\seqext_u^n,\Phi_u)$ can only hold if $\varepsilon_{n,\seqext_u^{n-1}}$ picks the default value $\seqext_u^n=\seqext_u^{n-1}$. Similarly we have $\seqext_u^{n-1}=\seqext_u^{n-2}$ and so on, so $\EPSs$ just returns the initial value $u$ at each step.

So how does the program justify selecting $u$ at point $n$, given that it has already chosen $u$ at $n-1$? We see that the selection function $\varepsilon_{n,u}$ always sets $\pair{\seqext_u^n,\seqfunc_u^n}=\pair{u,f_{|u_n|}}$ (where the $f_i$ are defined as in (\ref{defn-fx})), unless the outcome $Q_n(u,f_{|u_n|})=\Psi_u$ is lexicographically less than $u$ at point $n$, in which case it must check that $\bad(\Psi_u,f_{|u_n|}(\Psi_u))$ is false. But $f_{|u_n|}(\Psi_u)=J_n(\Psi_u,f_{|u_n|-1})$ by (\ref{defn-fx}) which checks the final outcome of $\EPSs$ given the sequence
\begin{equation*}(u,\seqfunc_u^0),\ldots,(u,\seqfunc_n^{n-1}),(\Psi_u,f_{|u_n|-1}) \ \ (\ast)\end{equation*}
Now in the computation of $\EPSs$ the functionals $\Omega$, $\Phi$, $\Psi$ only ever look at the first $i$ values of $\seqext^{i-1}$. Therefore we propose that because $\initSeg{\Psi_u}{n}=\initSeg{u}{n}$ (and $|u_n|-1=|(\Psi_u)_n|$) we can identify $(\ast)$ with the outcome of $\EPSs$ given the sequence
\begin{equation}(\Psi_u,\seqfunc_{\Psi_u}^0),\ldots,(\Psi_u,\seqfunc_{\Psi_u}^{n-1}),(\Psi_u,f_{|(\Psi_u)_n|})\end{equation}
which by our previous argument can be viewed as the outcome of running our algorithm with initial value $\Psi_u$ instead of $u$. In other words we make the identification $J_n(\Psi_u,f_{|u_n|-1})\sim\Phi_{\Psi_u}=\alg(\Psi_u)$, which explains why we must have $\neg\bad(\Psi_u,J_n(\Psi_u,f_{|u_n|-1}))$.

We claim that the algorithm $\Gamma$ obtained via $\EPSs$ has characteristics of an \emph{open recursion} procedure (see e.g. \cite{Berger(2004)}), computing $\alg(u)$ by internally computing values of $\alg(v)$ for $v$ lexicographically less than $u$. If we take $\seqext_u$ to be the constant sequence with value $u$, then our bound for $u$ is given by $\Gamma(u):=\Phi(\seqext_u,\seqfunc_u)=G(\varphi G)+1$ where now $G$ is a witness for $\Ram{\wq}$ on $\lts{u}$ and counterexample function $\lambda g\; . \; \seqfunc^{g0}(\initSeg{u}{g0}\ast (\fts{u}_{gi}))$. But by our previous argument we can identify $\seqfunc^{g0}(\initSeg{u}{g0}\ast (\fts{u}_{gi}))$ with $\Gamma(\initSeg{u}{g0}\ast (\fts{u}_{gi}))$. Thus it seems that $\Gamma$ is closely related to a functional $\tilde\Gamma$ defined, via {open recursion, by $\tilde\Gamma(u):=G(\varphi G)+1$ where $G$ is a witness for $\Ram{\wq}$ on the counterexample function $$\varphi:=\lambda g\; . \; \tilde\Gamma(\initSeg{u}{g0}\ast (\fts{u}_{gi})).$$

Of course, none of this precise - the identifications above are made very informally - and in particular we anticipate that the way our algorithm treats empty words would be more complex than a straightforward open recursion procedure. However, our purpose here is merely to provide via a casual argument some insight into how $\alg$ works.

It would be interesting to analyse the behaviour of our extracted algorithm in depth, to give a precise explanation of the way in which it computes bounds on bad sequences and compare this algorithm to those extracted using other methods. We leave this as an open problem.

\section{Final Comments}
\label{sec-comments}

We have used G\"{o}del's functional interpretation to produce a constructive version of Nash-Williams' minimal bad sequence proof of Higman's lemma. Our proof is relatively short and concise, and the combinatorial idea behind Nash-William's proof can be clearly seen in ours. Moreover, we can start to make sense of the operational behaviour of the extracted algorithm, at least on an informal level. We hope that this case study provides some insight into program extraction in infinitary combinatorics using the functional interpretation.

An obvious direction of future work is to better understand our realizer and give a more satisfactory description than that given in the previous section! One could potentially refine our realizer so that it is more intuitive and efficient, or alternatively construct a new realizer that directly interprets the functional interpretation of the minimal bad sequence argument and compare how it behaves to the one given here. It would also be instructive to formalise our program extraction in a theorem prover, and actually run the algorithm $\Gamma$ on some concrete WQOs to analyse its behaviour.

We close with the remark that the ideas in this article could be extended to solve the functional interpretation of the \emph{general} minimal bad sequence construction, and thereby extract programs from more complex proofs that use this construction, such as Kruskal's theorem. While our focus in this article was on the qualitative aspects of program extraction, it is natural to ask whether one could obtain useful \emph{quantitative} information from the analysis of proofs in this area of combinatorics. Bounds for Higman's lemma on a finite alphabet have already been produced using more direct methods e.g. \cite{Cichon(98)}, but it would be interesting to see if any useful constructive information could be extracted in the general case or for related theorems, through the formal analysis of proofs. \\

\noindent\textbf{Acknowledgements.} This work was supported by an EPSRC Doctoral Training Grant. The author thanks Paulo Oliva for suggesting this project and for reading an earlier draft of this article, and the anonymous referees for corrections and several useful comments.

\bibliographystyle{eptcs}

\bibliography{dblogic}

\end{document}

%% file: preamble0.tex
\newtheorem{theorem}{Theorem}

\newtheorem{lemma}[theorem]{Lemma}

\newtheorem{definition}[theorem]{Definition}
\newtheorem{remark}[theorem]{Remark}

\def\NN{\mathbb{N}}

\newcommand{\PAomega}{{\sf PA}^\omega}

\newcommand{\DC}{{\sf DC}}

\newcommand{\LEP}{{\sf LEP}}

\newcommand{\systemT}{{\sf T}}

\newcommand{\WQO}[1]{{\sf WQO}[{#1}]}

\newcommand{\MB}[1]{{\sf MB}[{#1}]}
\newcommand{\Ram}[1]{{\sf MonSeq}[{#1}]}

\newcommand{\EPS}[4]{{\sf EPS}^{#1}_{#4}({#2})({#3})}
\newcommand{\EPSs}{{\sf EPS}}

\newcommand{\pair}[1]{\langle #1 \rangle}

\newcommand{\initSeg}[2]{[{#1}]({#2})}
\newcommand{\ext}[1]{\widehat{#1}}

\newcommand{\bad}{\theta}

\newcommand{\ft}[1]{\widetilde{{#1}}}
\newcommand{\lt}[1]{\bar{{#1}}}
\newcommand{\fts}[1]{\tilde{{#1}}}
\newcommand{\lts}[1]{\bar{{#1}}}

\newcommand{\wq}{{X}}
\newcommand{\wwq}{{X^\ast}}
\newcommand{\e}{{x}}
\newcommand{\es}{x'}
\newcommand{\eseq}{{x}}
\newcommand{\eseqs}{{x}}

\newcommand{\ws}{y}

\newcommand{\wseq}{{u}}
\newcommand{\wseqs}{{v}}

\newcommand{\twseqs}{{\tilde{\wseqs}}}

\newcommand{\alg}{{\Gamma}}

\newcommand{\sbs}{\prec}
\newcommand{\sbw}{\unlhd}
\newcommand{\psbw}{\lhd}
\newcommand{\func}{{g}}

\newcommand{\seqext}{{\bf p}}
\newcommand{\seqfunc}{{\bf f}}

\newcommand{\sft}[2]{{\sf J}_{#2}{#1}}

\newcommand{\is}[1]{[{#1}]}
\newcommand{\nd}[3]{|{#1}|^{#2}_{#3}}

\newcommand{\smin}[1]{\sbw_{#1}}
\newcommand{\psmin}[1]{\psbw_{#1}}